\documentclass{llncs}

\usepackage{graphicx}
\usepackage{url}

\usepackage{float}
\usepackage{algorithm}
\usepackage{algorithmicx}     
\usepackage{algcompatible}    

\usepackage{time}
\usepackage{datetime}
\usepackage{amsmath}
\usepackage{amssymb}
\usepackage{lscape}
\usepackage{time}
\usepackage{datetime}

\usepackage{color}
\makeatletter 
\let\c@lofdepth\relax 
\let\c@lotdepth\relax 
\makeatother 
\usepackage[raggedright,small,sf,SF,hang]{subfigure}
\usepackage{caption}

\usepackage[toc,page,title,titletoc]{appendix}

\usepackage{float}
\usepackage{wrapfig}

\usepackage{extarrows}  

\usepackage{stmaryrd}   

\begin{document}



\renewcommand{\baselinestretch}{0.92} 
\newcommand{\clpr}{CLP(${\cal R}$)}
\newcommand{\np}{${\cal NP}$~}
\newcommand{\TR}{${\cal TR}$}
\newcommand{\CTR}{${\cal CTR}$}
\newcommand{\SR}{${\cal SR}$~}
\newcommand{\del}{{\em delta}}
\newcommand{\nil}{~{\rm nil}~}
\newcommand{\I}{{\cal I}~}

\newcommand{\ignore}[1]{}

\newcommand{\stuff}[1]{
        \begin{minipage}{6in}
        {\tt \samepage
        \begin{tabbing}
        \hspace{5mm} \= \hspace{5mm} \= \hspace{5mm} \= \hspace{5mm} \= \hspace{5mm} \= \hspace{5mm} \=
\hspace{5mm} \= \hspace{5mm} \= \hspace{5mm} \= \hspace{5mm} \= \hspace{5mm} \= \hspace{5mm} \= \hspace{5mm} \= \kill
        #1
        \end{tabbing}
       }
        \end{minipage}
}

\newcommand{\mystuff}[1]{
        \begin{minipage}[b]{6in}
        {\tt \samepage
        \begin{tabbing}
        \hspace{5mm} \= \hspace{5mm} \= \hspace{5mm} \= \hspace{5mm} \= \hspace{5mm} \= \hspace{5mm} \=
\hspace{5mm} \= \hspace{5mm} \= \hspace{5mm} \= \hspace{5mm} \= \hspace{5mm} \= \hspace{5mm} \= \hspace{5mm} \= \kill
        #1
        \end{tabbing}
       }
        \end{minipage}
}

\newcommand{\Rule}[2]{\genfrac{}{}{0.5pt}{}
{{\setlength{\fboxrule}{0pt}\setlength{\fboxsep}{3mm}\fbox{$#1$}}}
{{\setlength{\fboxrule}{0pt}\setlength{\fboxsep}{3mm}\fbox{$#2$}}}}

\newcommand{\rct}{\mbox{\it clp}}
\newcommand{\lif}{\ {\tt:\!\!-} \ }
\newcommand{\ctbl}{\,{\mbox{\footnotesize$|$}\tt-\!\!\!\!>}\,}
\newlength{\colwidth}
\setlength{\colwidth}{.47\textwidth}


\newcommand{\eqdef}{\stackrel{\rm def}{=}}
\newcommand{\sys}{\mbox{\bf sys}}
\newcommand{\var}{\mbox{\frenchspacing \it var}}
\newcommand{\lfp}{\mbox{\it \frenchspacing lfp}}
\newcommand{\gfp}{\mbox{\it \frenchspacing gfp}}
\newcommand{\assert}{\mbox{\it \frenchspacing assert}}
\newcommand{\buffer}{\mbox{\it \frenchspacing Buffer}}
\newcommand{\clpif}{\mbox{\tt :-}}


\newtheorem{myproof}{Proof Outline}

\newcommand{\QED}{\nolinebreak\hskip 1em
        \framebox[0.5em]{\rule{0ex}{1.0ex}}}

\newcommand{\close}{{\frenchspacing close }}
\newcommand{\rules}{{\frenchspacing rules }}
\newcommand{\lhs}{{\frenchspacing lhs }}
\newcommand{\rhs}{{\frenchspacing rhs }}
\newcommand{\wrt}{{\frenchspacing wrt. }}
\newcommand{\fold}{\mbox{\it\frenchspacing old}}
\newcommand{\exactunfold}{\mbox{\it\frenchspacing exactunfold}}

\newcommand{\mgu}{{\frenchspacing mgu }}
\newcommand{\size}{{\frenchspacing size }}
\newcommand{\obs}{{\frenchspacing obs }}
\newcommand{\trace}{{\frenchspacing trace }}

\newcommand{\xxx}{\mbox{\Large $\vartriangleright$}}

\newcommand{\undeniable}{\mbox{\Large $\vartriangleright$}\hspace{-8pt}\raisebox{2pt}{\tiny u}~~}
\newcommand{\inevitable}{\mbox{\Large $\vartriangleright$}\hspace{-8pt}\raisebox{2pt}{\tiny i}~~\,}

\newcommand{\reachable}{\mbox{\Large $\vartriangleright$}\hspace{-9pt}\raisebox{1pt}{\tiny *}~~\,}

\newcommand{\lbr}{\mbox{$[\![$}}
\newcommand{\rbr}{\mbox{$]\!]$}}
\newcommand{\ct}[1]{\mbox{\lbr$\vec{#1}$\rbr}}
\newcommand{\cs}[1]{\mbox{\lbr#1\rbr}}
\newcommand{\clp}[1]{{\it \frenchspacing clp}(\mbox{$#1$})}
\newcommand{\pp}[1]{\mbox{{\color{blue}{$\langle$#1$\rangle$}}}}

\newcommand{\ao}{\mbox{${\cal A}$}}
\newcommand{\co}{\mbox{${\cal C}$}}
\newcommand{\po}{\mbox{${\cal P}$}}

\newcommand{\tA}{{\tilde{A}}}
\newcommand{\tB}{{\tilde{B}}}
\newcommand{\tC}{{\tilde{C}}}

\newcommand{\A}{\mbox{$\cal A$}}
\newcommand{\B}{\mbox{$\cal B$}}
\newcommand{\D}{\mbox{$\cal D$}}
\newcommand{\E}{\mbox{$\cal E$}}
\newcommand{\M}{\mbox{$\cal M$}}
\newcommand{\V}{\mbox{$\cal V$}}
\newcommand{\unfold}{\mbox{\sc unfold}}

\newcommand{\G}{\mbox{$\cal G$}}
\newcommand{\Gone}{\mbox{${\cal G}_{\!1}$}}
\newcommand{\Gi}{\mbox{${\cal G}_{\!i}$}}
\newcommand{\Gn}{\mbox{${\cal G}_{\!n}$}}

\newcommand{\GL}{\mbox{${\cal G}_{\!L}$}}
\newcommand{\GR}{\mbox{${\cal G}_{\!R}$}}

\newcommand{\HH}{\mbox{$\cal H$}}
\newcommand{\HHone}{\mbox{${\cal H}_{\!1}$}}
\newcommand{\HHi}{\mbox{${\cal H}_{\!i}$}}
\newcommand{\HHj}{\mbox{${\cal H}_{\!j}$}}
\newcommand{\HHn}{\mbox{${\cal H}_{\!n}$}}
\newcommand{\HHm}{\mbox{${\cal H}_{\!m}$}}

\newcommand{\Mone}{\mbox{${\cal M}_{\!1}$}}
\newcommand{\Mtwo}{\mbox{${\cal M}_{\!2}$}}
\newcommand{\MMi}{\mbox{${\cal M}_{\!i}$}}
\newcommand{\MMn}{\mbox{${\cal M}_{\!n}$}}
\newcommand{\MMnone}{\mbox{${\cal M}_{\!n-1}$}}

\newcommand{\PsiL}{\mbox{$\Psi_{\!L}$}}
\newcommand{\PsiR}{\mbox{$\Psi_{\!R}$}}
\newcommand{\PsiB}{\mbox{$\Psi_{\!B}$}}

\newcommand{\Af}{A_{\!f}}
\newcommand{\Ax}{A_{\!1}}
\newcommand{\Nf}{N_{\!f}}
\newcommand{\Xf}{X_{\!f}}
\newcommand{\Yf}{Y_{\!f}}
\newcommand{\Hxxx}{H_{\!f}}

\newcommand{\cts}{\mbox{$\mapsto$}}
\newcommand{\sepimp}{\mbox{$-\!*$}}

\newcommand{\arr}[3]{\mbox{$\langle \mbox{#1,#2,#3} \rangle$}}
\newcommand{\triple}[3]{\langle{#1},{#2},{#3}\rangle}
\newcommand{\aquadruple}[4]{\langle{#1},{#2},{#3},{#4}\rangle}
\newcommand{\baeq}[2]{\mbox{$=_{\[#1 .. #2\]}$}}

\newcommand{\ite}[3]{\mbox{\textit{ite}}({#1},{#2},{#3})}
\newcommand{\pred}[1]{\mbox{\textit{#1}}}
\newcommand{\ptab}{~~~~}

\newlength{\vitelen}
\settowidth{\vitelen}{\mbox{\textit{ite}}(}
\newcommand{\vite}[3]{\begin{array}[t]{l}\mbox{\textit{ite}}({#1},\\
\hspace{\vitelen}{#2},\\
\hspace{\vitelen}{#3}
    \end{array}}

\newcommand{\Hfx}{H_{\!f}}
\newcommand{\Hx}{H_{\!1}}
\newcommand{\Hxx}{H_{\!2}}
\newcommand{\Pf}{P_{\!f}}
\newcommand{\Pz}{P_{\!0}}
\newcommand{\Px}{P_{\!1}}

\newcommand{\Ix}{I_{\!1}}
\newcommand{\Ixx}{I_{\!2}}
\newcommand{\Jf}{J_{\!f}}
\newcommand{\Jx}{J_{\!1}}
\newcommand{\Jxx}{J_{\!2}}

\newcommand{\witness}{\mbox{$\omega$}}
\newcommand{\transformer}{\mbox{$\Delta$}}

\newcommand{\func}[1]{\mbox{\textsf{#1}}}

\floatstyle{boxed}
\restylefloat{figure}

%

\renewcommand\labelitemi{$\bullet$}
\renewcommand\labelitemii{\normalfont\bfseries --}

\renewcommand\floatpagefraction{.9}
\renewcommand\topfraction{.9}
\renewcommand\bottomfraction{.9}
\renewcommand\textfraction{.1}   
\setcounter{totalnumber}{50}
\setcounter{topnumber}{50}
\setcounter{bottomnumber}{50}

\raggedbottom

\makeatletter
\newcommand{\manuallabel}[2]{\def\@currentlabel{#2}\label{#1}}
\makeatother

\newcounter{chapcount}
\newcommand{\chapcountreset}{\setcounter{chapcount}{0}}
\chapcountreset

\newcounter{excount}
\newcommand{\exreset}{\setcounter{excount}{0}}
\exreset
\newcommand{\newexample}[1]{\addtocounter{excount}{1}
{\vspace{2mm} \noindent \mbox{{\scriptsize EXAMPLE} \examplecount~\emph{#1}:}}}

\newcommand{\examplecount}{\arabic{excount}}

\newcommand{\exlabel}[1]{\manuallabel{#1}{\examplecount}}

\newcommand{\todo}[1]{
\noindent \framebox{\textbf{ #1}}
\newline
}

\newcommand{\algorithmicinput}{\textbf{Input:~}}
\newcommand{\algorithmicoutput}{\textbf{Output:~}}
\newcommand{\algorithmicglobal}{\textbf{Globally:~}}
\newcommand{\algorithmicfunction}{\textbf{Function}\ }
\newcommand{\algorithmicfunctionend}{\textbf{EndFunction}\ }
\newcommand{\memoed}{\mbox{\func{memoed}}}
\newcommand{\memoize}{\mbox{\func{memoize}}}

\newcommand{\memotable}{\mbox{\pred{Table}}}
\newcommand{\wpc}{\mbox{$\func{pre}$}}
\newcommand{\pre}[2]{\mbox{$\func{pre}(#1, #2)$}}

\newcommand{\interp}{\mbox{$\pred{Intp}$}}
\newcommand{\emanate}{\mbox{\func{outgoing}}}
\newcommand{\absiteration}{\mbox{\func{loop\_end}}}
\newcommand{\step}{\mbox{\func{TransStep}}}
\newcommand{\compress}{\mbox{\func{JoinVertical}}}
\newcommand{\join}{\mbox{\func{JoinHorizontal}}}

\newcommand{\program}{\mbox{$\cal P$}}
\newcommand{\N}{\mbox{$\cal N$}}

\newcommand{\Assign}{\mbox{:=}}
\newcommand{\pair}[2]{\langle{#1},{#2}\rangle}
\newcommand{\tuple}[3]{\langle{#1},{#2},{#3}\rangle}
\newcommand{\quadruple}[4]{[{#1},{#2},{#3},{#4}]}
\newcommand{\fivetuple}[5]{[{#1},{#2},{#3},{#4},{#5}]}

\newcommand{\sat}{{\small \textsf{SAT}}}
\newcommand{\unsat}{\textsf{UNSAT}}
\newcommand{\smt}{{\small \textsf{SMT}}}
\newcommand{\por}{{\small \textsf{POR}}}
\newcommand{\dpor}{{\small \textsf{DPOR}}} 
\newcommand{\pdpor}{{\small \textsf{PDPOR}}} 
\newcommand{\si}{{\small \textsf{SI}}}
\newcommand{\cegar}{{\small \textsf{CEGAR}}}
\newcommand{\al}{{\small \textsf{AL}}}
\newcommand{\rcsp}{{\small \textsf{RCSP}}}
\newcommand{\wcet}{{\small \textsf{WCET}}}
\newcommand{\ilp}{{\small \textsf{ILP}}}
\newcommand{\ipet}{{\small \textsf{IPET}}}
\newcommand{\cfg}{{\small \textsf{CFG}}}
\newcommand{\saturn}{{\small \textsf{SATURN}}}
\renewcommand{\dag}{{\small \textsf{DAG}}}

\newcounter{pppcount}
\newcommand{\pppreset}{\setcounter{pppcount}{0}}
\newcommand{\ppp}{\refstepcounter{pppcount}\mbox{{\color{blue}$\langle\arabic{pppcount}\rangle$}}}

\renewcommand{\note}[1]{\marginpar{\color{red}{#1}}}

\newcommand{\resource}{\mbox{\textsf r}}
\newcommand{\timing}{\mbox{\textsf t}}

\newcommand{\summarize}{\mbox{\cal S}}

\newcommand{\For}{\mbox{\textsf{for}}}
\newcommand{\Else}{\mbox{\textsf{else}}}
\newcommand{\If}{\mbox{\textsf{if}}}

\newcommand{\myiff}{\mbox{\texttt{iff}}}

\newcommand{\void}{\mbox{\textsf{void}}}
\newcommand{\assume}[1]{\mbox{\textsf{assume(#1)}}}
\newcommand{\assign}[2]{\mbox{\textsf{#1 := #2}}}

\newcommand{\exec}{\mbox{\textsf{exec}}}

\newcommand{\transition}[3]{#1~\xlongrightarrow[]{#3}~#2}
\newcommand{\shorttransition}[2]{#1~\xrightarrow[]{}~#2}
\newcommand{\trans}{\longrightarrow}
\newcommand{\shorttrans}{\rightarrow}
\newcommand{\translabel}[1]{\xlongrightarrow[]{#1}}
\newcommand{\loc}{\mbox{$\ell$}}
\newcommand{\locations}{\mbox{${\cal L}$}}

\newcommand{\transsystem}{\mbox{$\mathcal{P}$}}
\newcommand{\newtranssystem}{\mbox{$\mathcal{G}$}}

\newcommand{\symstate}{\mbox{$s$}}
\newcommand{\pci}[1]{\mbox{$\loc_{#1}$}}
\newcommand{\next}{\mbox{$\rightarrow$}}
\newcommand{\pc}{\mbox{\loc}}
\newcommand{\pcend}{\mbox{$\loc_{\textsf{end}}$}}
\newcommand{\pcerror}{\mbox{$\loc_{\textsf{error}}$}}
\newcommand{\pcstart}{\mbox{$\loc_{\textsf{start}}$}}
\newcommand{\pathcond}{\mbox{$\Pi$}}
\newcommand{\store}{\mbox{$\sigma$}}
\newcommand{\pathcondbar}{\mbox{$\overline{\Pi}$}}
\newcommand{\storebar}{\mbox{$\overline{h}$}}
\newcommand{\symstatebar}{\mbox{$\overline{\symstate}$}}
\newcommand{\mapstatetoformula}[1]{\mbox{$\llbracket {#1} \rrbracket$}}

\renewcommand{\path}{\mbox{$\theta$}}

\newcommand{\typevar}{\mbox{\emph{Vars}}}
\newcommand{\typesymvar}{\mbox{\emph{SymVars}}}
\newcommand{\typeop}{\mbox{\emph{Ops}}}
\newcommand{\typefo}{\mbox{\emph{FO}}}
\newcommand{\typeterms}{\mbox{\emph{Terms}}}
\newcommand{\typestate}{\mbox{\emph{States}}}
\newcommand{\typesymbstate}{\mbox{\emph{SymStates}}}
\newcommand{\typesympath}{\mbox{\emph{SymPaths}}}
\newcommand{\true}{\mbox{\frenchspacing \it true}}
\newcommand{\false}{\mbox{\frenchspacing \it false}}
\newcommand{\typebool}{\mbox{\emph{Bool}}}
\newcommand{\typeint}{\mbox{\emph{Int}}}
\newcommand{\typenat}{\mbox{\emph{Nat}}}
\newcommand{\typekeys}{\mbox{$\mathcal{K}$}}
\newcommand{\typevoid}{\mbox{\emph{Void}}}

\newcommand{\eval}[2]{\llbracket {#1} \rrbracket_{#2}}
\newcommand{\define}{\mbox{~$\triangleq$~}}
\newcommand{\unknown}{\mbox{\textsf{$\cdot$}}}

\newcommand{\Intpsymbol}{\mbox{$\overline{\Psi}$}}
\newcommand{\InvariantFunc}{\mbox{\textsf{invariant}}}
\newcommand{\InvariantSym}{\mbox{$\mathcal{I}$}}
\newcommand{\ConflictSym}{\mbox{$\mathcal{C}$}}
\newcommand{\ContextSym}{\mbox{$\mathcal{O}$}}
\newcommand{\modifies}{\mbox{\textsf{\textsc{Modifies}}}}
\newcommand{\havoc}{\mbox{\textsf{\textsc{Havoc}}}}
\newcommand{\getvars}{\mbox{\textsf{var}}}

\newcommand{\state}[1]{\mbox{{\small \textsf{#1}}}}
\newcommand{\safety}{\mbox{$\psi$}}
\newcommand{\cons}{\mbox{$\phi$}}
\newcommand{\target}{\mbox{$\gamma$}}
\renewcommand{\solution}{\mbox{$\gamma$}}
\newcommand{\solutions}{\mbox{$\Gamma$}}
\newcommand{\id}[2]{\mbox{\textsf{Id}($#1, #2$)}}

\renewcommand{\If}{\mbox{\textbf{if}}}
\newcommand{\Endif}{\mbox{\textbf{endif}}}
\newcommand{\Return}{\mbox{\textbf{return}}}
\newcommand{\Then}{\mbox{\textbf{then}}}
\renewcommand{\Else}{\mbox{\textbf{else}}}
\newcommand{\Foreach}{\mbox{\textbf{foreach}}}
\newcommand{\While}{\mbox{\textbf{while}}}
\newcommand{\Do}{\mbox{\textbf{do}}}
\newcommand{\Endfor}{\mbox{\textbf{endfor}}}
\newcommand{\Endwhile}{\mbox{\textbf{endwhile}}}

\floatsep 2mm plus 1mm minus 0mm

%
\textfloatsep 2mm plus 1mm minus 0mm

\setlength{\intextsep}{2pt plus 2pt minus 0pt}
\belowcaptionskip 2pt plus 0pt minus 2pt
\abovecaptionskip 2pt plus 0pt minus 2pt


\title{A Framework to Synergize Partial Order Reduction with State Interpolation}

\author{Duc-Hiep Chu \and Joxan Jaffar}

\institute{National University of Singapore\\
\email{hiepcd,joxan@comp.nus.edu.sg}}

\maketitle

\begin{abstract}

We address the problem of reasoning about interleavings 
in safety verification of concurrent programs.
In the literature, there are two prominent techniques for pruning the search space.  
First, there are well-investigated \emph{trace-based} methods, 
collectively known as ``Partial Order Reduction (\por{})'', 
which operate by weakening the concept of a trace by abstracting the total 
order of its transitions into a partial order. 
Second, there is \emph{state-based} interpolation where a collection of
formulas can be generalized by taking into account the property to be verified.
Our main contribution is a framework
that \emph{synergistically} combines \por{}
with state interpolation so that the sum is more than its parts.

\end{abstract}

\section{Introduction}
\label{sec:por:intro}
\noindent 
We consider the \emph{state explosion problem} 
in safety verification of concurrent programs.
This is caused by the interleavings of transitions from different processes.  
In explicit-state model checking, a general approach to counter this explosion 
is Partial Order Reduction (\por)~(e.g., \cite{valmari89petri,peled93cav,godefroid96}).
This exploits the equivalence of interleavings of 
``independent'' transitions: two transitions are independent 
if their consecutive occurrences in a trace 
can be swapped without changing the final state.
In other words, \por-related methods prune away 
\emph{redundant} process interleavings in a sense
that, for each Mazurkiewicz~\cite{mazurkiewicz86}\footnote{We
remark that the concept of \por{} goes beyond the preservation of Mazurkiewicz traces, 
e.g.~\cite{valmari89petri}.
However, from a practical perspective,
it is safe to consider such form of pruning as a representative example of \por{}.}
trace equivalence class of
interleavings, if a representative has been checked, the remaining
ones are regarded as redundant.

On the other hand, \emph{symbolic execution} 
\cite{king76acm} is another method for program reasoning
which recently has made increasing impact
on software engineering research \cite{cadar11icse}.
The main challenge for symbolic execution 
is the exponential number of symbolic paths. 
The works \cite{mcmillan10cav,jaffar11rv} 
tackle successfully this fundamental problem by 
eliminating from the concrete model, on-the-fly, those facts which are \emph{irrelevant} or 
\emph{too-specific} for proving the unreachability of the error nodes.
 This learning phase consists of computing 
\emph{state-based interpolants} in a similar spirit 
to that of conflict clause learning in \sat~solvers. 

Now symbolic execution with state interpolation (\si{}) 
has been shown to be effective for verifying sequential programs.
In \si~\cite{mcmillan10cav,jaffar11rv},
a node at program point $\pc$ in the reachability tree can be pruned,
if its context is subsumed by the interpolant computed earlier for the
same program point $\pc$.  Therefore, even in the best case scenario,
the number of states explored by an \si~method
must still be at least the number of all \emph{distinct} program
points\footnote{Whereas \por-related methods do not suffer from this. Here we assume
that the input concurrent program has already been preprocessed
(e.g., by static slicing to remove irrelevant transitions, or by static block encodings) to reduce
the size of the transition system for each process.}.
However, in the setting of concurrent programs, exploring each distinct
global program point\footnote{The number of global points 
is the product of the numbers of local program points in all processes.} 
once might already be considered prohibitive.  
In short, symbolic execution with \si~\emph{alone} 
is not efficient enough for the verification of concurrent programs.

Recent work (e.g., \cite{StatefulSPIN08}) has
shown the usefulness of going \emph{stateful} in implementing 
a \por~method.  It directly follows that \si~can help to
yield even better performance.  In order to implement an efficient
stateful algorithm, we are required to come up with an abstraction
for each (concrete or symbolic) state.  Unsurprisingly,
\si~often offers us good
abstractions.

The above suggests that
\por~and \si~can be very much \emph{complementary} to each other.
In this paper, we propose a general framework
employing \emph{symbolic execution} in the 
exploration of the state space, 
while both \por~and \si~are exploited for pruning. 
\si{} and \por{} are combined synergistically as the concept of interpolation. 
Interpolation is essentially a form of learning where the 
completed search of a \emph{safe} subtree is then formulated as a
recipe, ideally a \emph{succinct} formula, 
for future pruning.  The key distinction of our interpolation framework is that each recipe
discovered by a node is \emph{forced} to be conveyed back to its ancestors, which
gives rise to pruning of larger subtrees.

In summary, we address the challenge:
``combining classic \por~methods with symbolic technique 
has proven to be difficult''~\cite{nec09cav},
especially in the context of \emph{software verification}.
More specifically, we propose an algorithm schema to combine 
\emph{synergistically} \por{} with state interpolation so that the sum is more than its
parts. However, we first need to formalize \por~wrt. a 
symbolic search framework with abstraction 
in such a way that: (1) \por~can be \emph{property driven} and (2) 
\por{}, or more precisely, the concept of persistent set,  
can be applicable for a set of states (rather than an individual state). 
While the main contribution is a theoretical framework (for space reason, 
we present the proofs of the two theorems in the Appendix),
we also indicate a potential for the development of advanced implementations.






\section{Related Work}
\label{sec:por:related}

Partial Order Reduction (\por) is a well-investigated technique
in model checking of concurrent systems.
Some notable early works are \cite{valmari89petri,peled93cav,godefroid96}.
Later refinements of \por{}, Dynamic~\cite{flanagan05popl} and
Cartesian~\cite{cartesian} \por~({\small \textsf{DPOR}} and {\small \textsf{CPOR}} respectively) 
improve traditional \por~techniques by
detecting collisions on-the-fly.  
Recently, \cite{abdulla14popl} has proposed the novel of concept of \emph{source sets},
optimizing the implementation for \dpor{}.
These methods, in general, often achieve better reduction than traditional techniques,
due to the more accurate detection of independent transitions.

Traditional \por~techniques \cite{valmari89petri,peled93cav,godefroid96} 
distinguish between liveness and safety properties.
\por~ has also been extended for symbolic model checking \cite{alur97cav}: 
a symbolic state can represent a number of concrete states.
These methods, however, are not applicable to safety verification 
of modern concurrent programs (written in mainstream APIs such as POSIX).
One important weakness of traditional \por~is that it is \emph{not sensitive}
wrt. different target safety properties. 
%
In contrast, recent works have shown that
property-aware reduction can be achieved by symbolic
methods using a general-purpose \sat/\smt{} solver \cite{nec08tacas,nec09cav,nec09fse,ICSE11}. 
Verification is often encoded as a formula which is \emph{satisfiable}
\myiff{} there exists an interleaving execution of the programs that
violates the property. Reductions happen inside the \sat ~solver through
the addition of learned clauses derived by conflict
analysis~\cite{grasp}.  This type of reduction
is somewhat similar to what we call \emph{state interpolation}. 
%


An important related work is \cite{nec09cav}, which is the first
to consider enhancing \por{}  with property driven pruning, via the use of 
an \smt{} solver.
Subsequently, there was a follow-up work \cite{nec09fse}.
In \cite{nec09cav}, they began with an \smt{} encoding of the underlying transition system,
and then they enhance this encoding with a concept of ``monotonicity''.
The effect of this is that traces can be grouped into equivalence classes,
and in each class, all traces which are \emph{not monotonic} will be considered
as \emph{unsatisfiable} by the \smt~solver.
The idea of course is that such traces are in fact redundant.
This work has demonstrated some promising results as most concurrency bugs in real
applications have been found to be \emph{shallow}. We note that
\cite{nec09cav} incidentally enjoyed some (weak) form of \si{} pruning, due to the similarity
between conflict clause learning and state interpolation. However, there the synergy 
between \por{} and \smt{}
is \emph{unclear}. We later demonstrate in Sec.~\ref{sec:por:experiment} that
such synergy in \cite{nec09cav} is indeed relatively poor.

There is a fundamental problem with scalability in \cite{nec09cav}, as
mentioned in the follow-up work \cite{nec09fse}:
``It will not scale to the entire concurrent program'' if we encode the 
whole search space as a single formula and submit it to an \smt~solver.

Let us first compare \cite{nec09cav} with our work.
Essentially, the difference is twofold.
First, in this paper, the theory for partial order reduction is \emph{property driven}.
In contrast, the monotonicity reduction of \cite{nec09cav} is not. In other words,
though property driven pruning is observed in \cite{nec09cav}, it is contributed mainly
by the conflict clauses learned, not from the monotonicity relation.
We specifically exemplify the power of property driven \por{} in the later sections.
Second, the encoding in \cite{nec09cav} is processed by a \emph{black-box}
\smt{} solver.  Thus important algorithmic refinements 
are not possible. Some examples:

\vspace{2mm}
\noindent $\bullet$
There are different options in implementing \si{}.
Specifically in this paper, we employ ``precondition'' computations.
Using a black-box solver, one has to rely on its 
fixed interpolation methods.


\vspace{2mm}
\noindent $\bullet$
Our approach is \emph{lazy} in a sense that our solver
is only required to consider \emph{one} symbolic path at a time;
in \cite{nec09cav} it is not the case. This matters most
when the program is unsafe and finding
counter-examples is relatively easy (there are many traces
which violate the safety property).

\vspace{2mm}
\noindent $\bullet$
In having a symbolic execution framework, one can direct the search process.
This is useful since the order in which state interpolants are generated 
does give rise to different reductions.
Of course, such manipulation of the search process is 
hard, if not impossible, when using a black-box solver.

\vspace{2mm}
In order to remedy the scalability issue of \cite{nec09cav},
the work \cite{nec09fse} adapted it to the setting of program testing.
In particular, \cite{nec09fse} proposed a concurrent trace program (CTP)
framework which employs both concrete execution and symbolic solving
to strike a balance between efficiency and scalability of an \smt-based
method.  However, when the input program is \emph{safe}, i.e., absence of bugs,
\cite{nec09fse} in general suffers from the same scalability issue as in \cite{nec09cav}.

We remark that, the new direction of \cite{nec09fse}, in avoiding the blow-up of the \smt{} solver, 
was in fact preceded by the work on under-approximation
widening (UW)~\cite{grumberg05popl}.
As with CTP, UW models a subset,
which will be incrementally enlarged, of all the possible
interleavings as an \smt~formula and submits it to an \smt~solver. 
In UW the scheduling decisions are also encoded as constraints, so that the
\emph{unsatisfiable core} returned by the solver can then be used to
further the search in probably a useful direction.  
This is the major contribution of UW.
However, an important point is that this furthering of the search
is a \emph{repeated} call to the solver, this time with a weaker formula;
which means that the problem at hand is now larger, having more
traces to consider.
On this repeated call, the work done for the original call is
thus \emph{duplicated}.


At first glance, it seems attractive and simple to encode the problem
compactly as a set of constraints and delegate the search process to a
general-purpose \smt{} solver.  However, there are some fundamental
disadvantages, and these arise mainly because it is hard to exploit
the semantics of the program to direct the search inside the solver.
This is in fact evidenced in the works mentioned above.

We believe, however, the foremost disadvantage of using a general-purpose solver
lies in the encoding of process interleavings.
For instance, even when a concurrent program 
has only \emph{one} feasible execution trace, the encoding formula
being fed to the solver is still of enormous size and
can easily choke up the solver. More
importantly, different from safety verification of sequential programs,
the encoding of interleavings (e.g., \cite{nec09cav} uses
the variable \emph{sel} to model which process is selected for executing) 
often hampers the normal derivations of succinct conflict clauses
by means of resolution in modern \smt{} solvers.
We empirically demonstrate the inefficiency 
of such approach in Sec.~\ref{sec:por:experiment}.



Another important related work is \cite{wachter13fmcad},
developed independently\footnote{Our work has been publicly available since 2012 in
forms of a draft paper and a Ph.D. thesis. For anonymity reason, we omit the details here.} by Wachter et al.
\cite{wachter13fmcad} follows a similar direction as in this paper: combining \por{} with 
a standard state interpolation algorithm, which is often referred to as the {\small {\textsf IMPACT}} 
algorithm~\cite{mcmillan10cav}. 
Nevertheless, it is important to note that the framework presented in this paper subsumes
\cite{wachter13fmcad}. While this paper proposes the novel concept of
Property Driven \por{} before combining it with the state interpolation algorithm, 
\cite{wachter13fmcad} exploits directly the concept of ``monotonicity'' as in \cite{nec09cav},
thus their \por{} part does not give rise to property driven pruning.

%


\section{Background}
\label{sec:por:prelim}
We consider a concurrent system composed of a finite number of threads or processes
performing atomic operations on shared variables. 
%
Let $P_i$ ($1 \leq i \leq n$) be a process with the set
$trans_i$ of transitions.
For simplicity, assume that $trans_i$ contains no cycles.

We also assume all processes have disjoint sets of transitions. 
Let $\mathcal{T} = \cup_{i=1}^{n} trans_i$ be the set of all transitions.
Let $V_i$ be the set of local variables 
of process $P_i$, and $V_{shared}$ the set of shared variables 
of the given concurrent program.
Let $pc_i \in V_i$ be a special variable representing the process program counter, 
and the tuple $\langle pc_1, pc_2 \cdots, pc_n \rangle$ 
represent the global program point.
Let $\typesymbstate$ be the set of all global symbolic states
of the given program where $s_0 \in \typesymbstate$ is the initial state. 
A state $s \in \typesymbstate$ comprises two parts: 
its {\em global program point} $\pc$, also denoted by $\texttt{pc}(s)$,
which is a tuple of local program counters, 
and its {\em symbolic constraints} $\mapstatetoformula{\symstate}$
over the program variables. In other words, we denote a state $s$ by
$\pair{\texttt{pc}(s)}{\mapstatetoformula{\symstate}}$.

We consider the {\em transitions} of states
induced by the program. 
Following~\cite{godefroid96}, we only pay attention to \emph{visible} transitions.
A (visible) transition $t^{\{i\}}$ pertains to some process $P_i$.
It transfers process $P_i$ from control location $\pc_{1}$ to $\pc_{2}$.
In general, the application of $t^{\{i\}}$ is guarded by some 
condition \texttt{cond} (\texttt{cond} might be just \texttt{true}).
At some state $s \in \typesymbstate$, 
when the $i^{th}$ component  of $\texttt{pc}(s)$, namely $\texttt{pc}(s)[i]$,
equals $\pc_{1}$, we say that $t^{\{i\}}$ is 
\emph{schedulable}\footnote{This concept is not standard in traditional \por{},
we need it here since we are dealing with symbolic search.}
at $s$.
And when $\symstate$
 satisfies the guard \texttt{cond}, denoted by
$\symstate \models \texttt{cond}$, 
we say that $t^{\{i\}}$ is \emph{enabled} at $s$.
For each state $s$, let $Schedulable(s)$ and $Enabled(s)$ denote the set of transitions 
which respectively are schedulable at $s$ and enabled at $s$.
A state $s$, where $Schedulable(s) = \emptyset$, is called a \emph{terminal state}.

Let $s \stackrel{t}{\rightarrow} s'$ denote transition step 
from $s$ to $s'$ via transition $t$. This step is possible only if
$t$ is \emph{schedulable} at $s$.
We assume that
the effect of applying an enabled transition $t$ on a state $s$ to 
arrive at state $s'$ is well-understood.
In our symbolic execution framework,
executing a schedulable but not enabled 
transition results in an \emph{infeasible} state.
A state $s$ is called {\em infeasible} if $\mapstatetoformula{\symstate}$ is unsatisfiable.
For technical reasons needed below, we shall allow schedulable transitions emanating
from an infeasible state; it follows that the destination state must also be \emph{infeasible}.

\ignore{
\newexample{}
Consider two processes $P_1, P_2$: 
$P_1$ simply awaits for $x=0$, while $P_2$ increments $x$.
So each has one transition to transfer (locally) 
from control location \pp{0} to \pp{1}.
Assume that initially $x=0$, i.e., the initial state $s_0$ is 
$\pair{\pp{0,0}}{x=0}$.
Running $P_2$ first we have the transition from state
$\pair{\pp{0,0}}{x=0}$
to
state $\pair{\pp{0,1}}{x=1}$. 
From here, we note that the transition from $P_1$ is now not enabled
even though it is schedulable. 
If applied, it produces an infeasible (and terminal) state 
$\pair{\pp{1,1}}{x=1 \wedge 1=0}$.
Note that $(x=1 \wedge 1=0) \equiv false$.

On the other hand, running $P_1$ first, we have the transition
from 
$\pair{\pp{0,0}}{x=0}$
to $\pair{\pp{1,0}}{x=0}$.
We may now have a subsequent transition step to
$\pair{\pp{1,1}}{x=1}$,
which is a feasible terminal state. 
}

For a sequence of transitions $w$ (i.e., $w \in \mathcal{T}^*$), 
$Rng(w)$ denotes the set of transitions that appear in $w$. 
Also let $\mathcal{T}_{\pc}$ denote 
the set of all transitions which are schedulable somewhere 
after global program point $\pc$. We note here that the 
schedulability of a transition at some state $s$ only depends 
on the program point component of $s$, namely $\texttt{pc}(s)$.
It does not depend on the constraint component of s, namely $\mapstatetoformula{\symstate}$.
Given $t_1, t_2 \in \mathcal{T}$ we say $t_1$ can \emph{de-schedule} $t_2$ \myiff{}
there exists a state $\symstate$ such that both $t_1, t_2$ are schedulable at $\symstate$ but $t_2$ is 
not schedulable after the execution of $t_1$ from $\symstate$.

Following the above,
$s_1 \stackrel{t_1 \cdots t_m}{\Longrightarrow} s_{m+1}$ 
denotes a sequence of state transitions,
and we say that $s_{m+1}$ is reachable from $s_1$.
We call 
$s_1 \stackrel{t_1}{\rightarrow} s_2  \stackrel{t_2}{\rightarrow} \cdots  \stackrel{t_m}{\rightarrow} s_{m+1}$
a \emph{feasible} derivation from state $s_1$, 
\myiff{} $\forall~1 \leq i \leq m \bullet t_i$ is enabled at $s_i$.
As mentioned earlier, an \emph{infeasible} 
derivation results in an \emph{infeasible state}
(an infeasible state is still aware of its global program point). 
An infeasible state satisfies any safety property.

We define a \emph{complete execution} trace, or simply trace, 
$\rho$ as a sequence of transitions such that 
it is a derivation from $s_0$ and $s_0 \stackrel{\rho}{\implies} s_{f}$ and
$s_{f}$ is a terminal state. 
A trace is infeasible if it is an infeasible derivation from $s_0$.
If a trace is infeasible, then at some point, 
it takes a transition which is schedulable but is not enabled. 
From thereon, the subsequent states are infeasible states.

%

We say a given concurrent program is \emph{safe} 
wrt. a safety property $\safety$ if
$\forall s \in \typesymbstate ~\bullet$ if \emph{s is reachable from the initial state} $s_0$ 
then $s \models \safety$.
A trace $\rho$ is \emph{safe} wrt. $\safety$, 
denoted as $\rho \models \safety$, 
if all its states satisfy $\safety$. 

\vspace*{-2mm}
\subsection*{Partial Order Reduction (\por{}) vs. State-based Interpolation (\si{})}

We assume the readers are familiar with the traditional concept of \por{}.
Regarding state-based interpolation,  we follow the approach of~\cite{jaffar09cp,mcmillan10cav}.
Here our symbolic execution is depicted as a tree rooted at the initial
state $s_0$ and for each state $s_i$ therein,
the descendants are just the states obtainable by extending $s_i$ with a
feasible transition.

\ignore{
Consider one particular feasible path:
$s_0 \stackrel{t_1}{\rightarrow} s_1
\stackrel{t_2}{\rightarrow} s_2 \cdots s_m$.  A \emph{program point}
$\loc_i$ of $\symstate_i$
characterizes a point in the reachability tree in terms of all the
remaining possible transitions.  Now, this particular path is
\emph{safe} wrt. a safety property $\psi$ if for all $k$, $0 \leq k \leq m$, 
we have $s_k \models \psi$.  
A (state) interpolant at program point $\loc_i$, $0 \leq i \leq m$ is simply
a set of states $S_i$ containing $s_i$ such that for any state
$s'_i \in S_i$, $s'_i \stackrel{t_{i+1}}{\longrightarrow} s'_{i+1}
\stackrel{t_{i+2}}{\longrightarrow} s'_{i+2} \cdots s'_m$, it is
also the case that for all $k$, $i \leq k \leq m$, we have 
$s'_k \models \psi$.  This interpolant was
constructed at program point $\loc_i$ due to the one path.  Consider now all
paths from $s_0$ and with prefix $t_1, \cdots , t_{i-1}$.
Compute each of their interpolants.  
Finally, the interpolant for the subtree (at $s_i$) of
paths just considered is simply the intersection of all the
individual interpolants.  This notion of interpolant for a subtree
provides a weaker notion of \emph{subsumption}.
We can now prune a subtree in case its root is within the interpolant computed for a
previously encountered subtree of the same program point. 
}

\begin{definition}[Safe Root]
Given a transition system and an initial state $\symstate_0$,
let $\symstate$ be a feasible state reachable from $\symstate_0$. We say
$\symstate$ is a \emph{safe root} wrt. a safety property $\safety$, 
denoted $\bigtriangleup_{\safety}(\symstate)$, 
\emph{\myiff} all states $\symstate'$ reachable from $\symstate$ are safe wrt. $\safety$. 
\end{definition}

\begin{definition}[State Coverage]\label{def:background:state_pruning}
Given a transition system and an initial state $\symstate_0$ and
$\symstate_i$ and $\symstate_j$ are two symbolic states such that
(1) $\symstate_i$ and $\symstate_j$ are reachable from $\symstate_0$ and
(2) $\symstate_i$ and $\symstate_j$ share the same program point $\pc$,
we say $\symstate_i$ \emph{covers} $\symstate_j$ wrt. a safety property $
\safety$, denoted by $\symstate_i \succeq_{\safety} \symstate_j$,
\emph{\myiff}~$\bigtriangleup_{\safety}(\symstate_i)$ implies $\bigtriangleup_{\safety}(\symstate_j)$.
\end{definition}

The impact of state coverage relation is that if 
(1) $\symstate_i$ covers $\symstate_j$, and 
(2) the subtree rooted at $\symstate_i$ has
been traversed and proved to be safe, then the traversal 
of subtree rooted at $\symstate_j$ can be avoided.
In other words, we gain performance by \emph{pruning} the subtree at $\symstate_j$.
Obviously, if $\symstate_i$ naturally subsumes $\symstate_j$, i.e.,
$\symstate_j \models
\symstate_i$, then state coverage is
trivially achieved. In practice, however, 
this scenario does not happen often enough. 

\ignore{
Let us now introduce the concept of Craig 
interpolant~\cite{craig55interpolant}.

\begin{definition}[Interpolant] \label{def:interpolant}
  Given two first-order logic formulas $F$ and $G$ such that $F \models G$, then there
  exists  an {\em interpolant $H$} denoted as $\emph{\interp}(F,G)$,
  which is a first-order logic formula such that $F \models H$ and $H \models G$, 
  and each variable of $H$ is a variable of both $F$ and $G$.
\end{definition}
}

\begin{definition}[Sound Interpolant]\label{def:interpolant:sound}
Given a transition system and an initial state $\symstate_0$,
given a safety property $\safety$ and program point $\pc$, 
we say a formula $\Intpsymbol$ is 
a \emph{sound interpolant} for $\pc$, denoted by \emph{$\si(\pc,\safety)$},
if for all states $\symstate \equiv \pair{\pc}{\unknown}$ 
reachable from $\symstate_0$,  
$\symstate \models \Intpsymbol$
implies that $\symstate$ is a safe root.
\end{definition}

What we want now is to generate a formula $\Intpsymbol$ 
(called \emph{interpolant}), which
still preserves the safety of all states reachable from $\symstate_i$, 
but is weaker (more general)  than the original formula associated to the state
$\symstate_i$.
In other words, we should have $\symstate_i 
\models \si{}(\pc,\safety)$. 
We assume that this condition is always ensured by any implementation 
of state-based interpolation. The main purpose of using
$\Intpsymbol$ rather than the original formula associated to the
symbolic state $\symstate_i$ is to increase the
likelihood of subsumption. That is,
the likelihood of having $\symstate_j \models \Intpsymbol$
is expected to be much higher than the likelihood of
having $\symstate_j \models \symstate_i$.

In fact, the perfect interpolant should be the weakest precondition \cite{dijkstrawp}
computed for program point $\pc$ wrt. the transition system 
and the safety property $\safety$. 
We denote this weakest precondition
as $\func{wp}(\pc,\safety)$.
Any subsequent state 
$\symstate_j~\equiv~\pair{\pc}{\unknown}$
which has $\symstate_j$ 
stronger than this weakest precondition can be pruned.
However, in general, the weakest precondition is too computationally demanding.
An interpolant for the state $\symstate_i$ is indeed a formula which 
approximates the weakest precondition
at program point $\pc$ wrt. the transition system,
i.e., $\Intpsymbol\equiv\si{}(\pc,\safety)\equiv\interp(\symstate_i,\func{wp}(\pc,\safety))$.
A \emph{good} interpolant is one which closely approximates the weakest precondition
and can be computed efficiently.

The symbolic execution of a program can be augmented by annotating
each program point with its corresponding interpolants such that the
interpolants represent the sufficient conditions to preserve the
unreachability of any unsafe state. Then, the \emph{basic} 
notion of pruning with state interpolant can be defined as follows.

\begin{definition}[Pruning with Interpolant]
\label{def:subsumption}
  Given a symbolic state $\symstate \equiv \pair{\pc}{\unknown}$  such that
  $\pc$ is annotated with some interpolant $\Intpsymbol$, we say
  that $\symstate$ is \emph{pruned} by the interpolant $\Intpsymbol$ if
  $\symstate$ implies $\Intpsymbol$ (i.e.,
  $\symstate \models \Intpsymbol$).
\end{definition}


Now let us discuss the the effectiveness of \por{} and \si{} in pruning
the search space with an example.
For simplicity, we purposely make the example \emph{concrete}, i.e., states
are indeed concrete states.

\newexample{(Closely coupled processes)}\exlabel{por:ex:close}
See Fig.~\ref{por:fig:state_close}.
Program points are shown in angle brackets.
Fig.~\ref{por:subfig1:state_close} shows the 
control flow graphs of two processes.
Process 1 increments $x$ twice whereas 
process 2 doubles $x$ twice. 
The transitions associated with such actions  
and the safety property are depicted in the figure.
\por{} requires a full search tree while
Fig.~\ref{por:subfig2:state_close} shows
the search space explored by \si{}.
Interpolants are in curly brackets.
Bold circles denote pruned/subsumed states.

\begin{figure}[tbh]
\subfigure[Two Closely Coupled Processes]{
 \includegraphics[width=.40\textwidth]{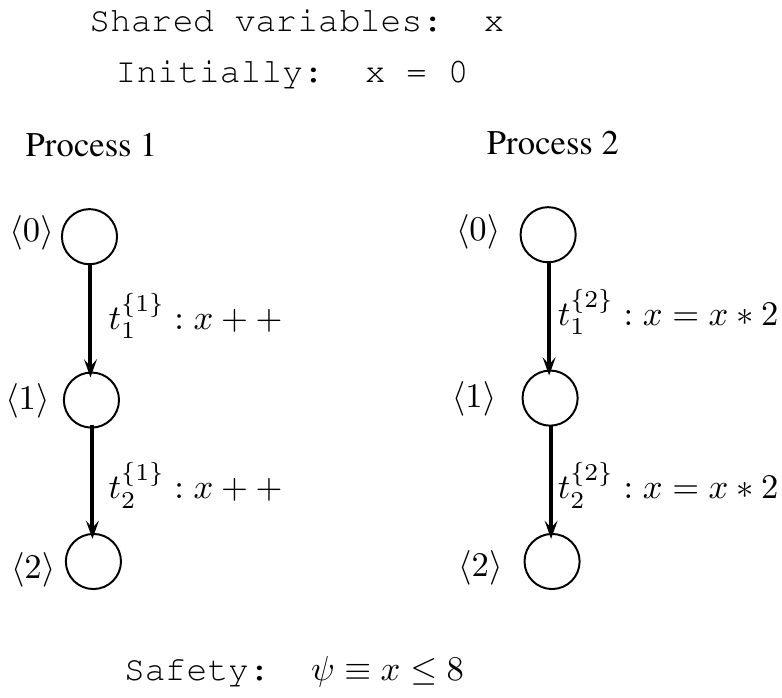}
\label{por:subfig1:state_close}
 }
 \subfigure[Search Space by \si{}] {
\includegraphics[width=.55\textwidth]{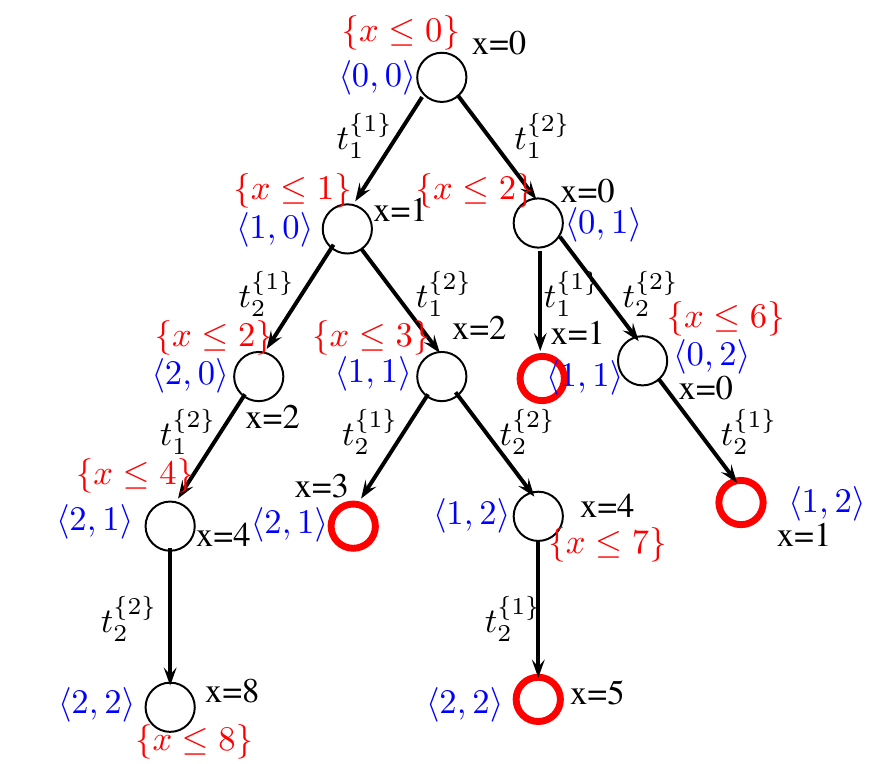}
\label{por:subfig2:state_close}
}
\vspace{-2mm}
\caption{Application of \si{} on 2 Closely Coupled Processes}
\label{por:fig:state_close}
\end{figure}

\noindent
Let us first attempt this example using \por.
It is clear that 
$t^{\{1\}}_{1}$ is \emph{dependent} with both $t^{\{2\}}_{1}$ and $t^{\{2\}}_{2}$. Also
$t^{\{1\}}_{2}$ is dependent with both $t^{\{2\}}_{1}$ and $t^{\{2\}}_{2}$.
Indeed, each of all the 6 execution traces in the search tree 
ends at a different concrete state.
As classic \por{} methods use the concept of \emph{trace equivalence} for pruning,
no interleaving is avoided:
those methods will enumerate the full search tree of 19 states (for space reasons, we omit it here).

Revisit the example using \si{}, where we use
the \emph{weakest preconditions} \cite{dijkstrawp}
as the state interpolants: the interpolant for a 
state is computed as the weakest precondition
to ensure that the state itself as well as all of its descendants are safe
(see Fig.~\ref{por:subfig2:state_close}).
We in fact achieve the best case scenario with it: whenever we come 
to a program point which has been examined before, subsumption happens.
The number of non-subsumed states is still of order $O(k^2)$ 
(where $k = 3$ in this particular example), assuming that we
generalize the number of local program points for each process to $O(k)$.
Fig.~\ref{por:subfig2:state_close} shows 9 non-subsumed states and 4 subsumed states.

\ignore{
Note the similarity between the tree in 
Fig.~\ref{por:subfig2:state_close}
and the tree in Fig.~\ref{por:subfig3:por_loose},
even though Ex.~\ref{por:ex:close}
contains more interfering transitions, 
comparing to Ex.~\ref{por:ex:loose}.
}

In summary, the above example shows that
\si{} might outperform \por{} when the component processes are
closely coupled. However, one can easily devise an example where
the component processes do not interfere with each other at all. Under
such condition \por{} will require only one trace to prove safety, while \si{}
is still (lower) bounded by the total number of global program points.
In this paper, we contribute by proposing a framework to combine 
\por{} and \si{} \emph{synergistically}.




\section{Property Driven POR (PDPOR)}
\label{sec:targetPOR}
``Combining classic \por~methods with symbolic
algorithms has been proven to be difficult''~\cite{nec09cav}. 
One fundamental reason is that the concepts of (Mazurkiewicz) 
equivalence and transition independence, 
which drive most practical \por~implementations, 
rely on the equivalence of two concrete states.
However, in symbolic traversal, we 
rarely encounter two equivalent symbolic states.


We now make the following definition which is
crucial for the concept of pruning and
will be used throughout this paper.

\begin{definition}[Trace Coverage]
Let $\rho_1, \rho_2$ be two traces of a concurrent program.
We say $\rho_1$ covers $\rho_2$ wrt. a safety property $\psi$, denoted as
$\rho_1 \sqsupseteq_{\psi}  \rho_2$, \emph{\myiff}
$\rho_1 \models \psi \rightarrow \rho_2 \models \psi$.
\end{definition}

Instead of using the concept of trace equivalence, from now on, 
we only make use of the concept of trace coverage.
The concept of trace coverage is definitely weaker than 
the concept of Mazurkiewicz equivalence. 
In fact, if $\rho_1$ and $\rho_2$ are (Mazurkiewicz) equivalent then
$\forall \psi \bullet \rho_1 \sqsupseteq_{\psi}  \rho_2 \wedge \rho_2 \sqsupseteq_{\psi}  \rho_1$. 
Now we will define a new and \emph{weaker} concept which therefore generalizes
the concept of transition independence. 

\ignore{
\begin{definition}[Semi-Commutative After A State]
For all $\theta, w_1, w_2 \in \mathcal{T}^*$ and for all 
$t_1, t_2 \in \mathcal{T}$ and for all property $\psi$
such that (1) $s_0 \stackrel{\theta}{\Longrightarrow} s$ is a feasible 
derivation; (2) 
$\theta w_1 t_1 t_2 w_2$ and $\theta w_1 t_2 t_1 w_2$ 
both are execution traces of the program,
we say $t_1$ \emph{semi-commutes} with $t_2$ \emph{after state} $s$ wrt. $\sqsupseteq_{\psi}$, 
denoted by $\langle s, t_1~\uparrow~t_2, \psi \rangle$ \emph{\myiff}
$\theta w_1 t_1 t_2 w_2 \sqsupseteq_{\psi} \theta w_1 t_2 t_1 w_2$.
\end{definition}
}

\begin{definition}[Semi-Commutative After A State]
For a given concurrent program, a safety property $\safety$,
and a 
derivation $s_0 \stackrel{\theta}{\Longrightarrow} s$,
for all $t_1, t_2 \in \mathcal{T}$ which cannot de-schedule each other,
we say $t_1$ \emph{semi-commutes} with 
$t_2$ \emph{after state} $s$ wrt. $\sqsupseteq_{\psi}$, 
denoted by $\langle s, t_1 \uparrow t_2, \psi \rangle$, \emph{\myiff}
for all $w_1, w_2 \in \mathcal{T}^*$ such that
$\theta w_1 t_1 t_2 w_2$ and $\theta w_1 t_2 t_1 w_2$ 
are execution traces of the program, then we have
$\theta w_1 t_1 t_2 w_2 \sqsupseteq_{\psi} \theta w_1 t_2 t_1 w_2$.
\end{definition}

From the definition, $Rng(\theta)$, $Rng(w_1)$, and $Rng(w_2)$ are pairwise disjoint.
Importantly, if $s$ is at program point $\pc$, 
we have $Rng(w_1) \cup Rng(w_2) \subseteq \mathcal{T}_{\pc} \setminus \{t_1, t_2\}$.  
We observe that wrt. some $\safety$, 
if all important events, those have to do with the safety of the system,
 have already happened in the prefix $\theta$, 
the ``semi-commutative'' relation is trivially satisfied. 
On the other hand, the remaining transitions might 
still interfere with each other (but not the safety), and do 
not satisfy the independent relation.  


The concept of ``semi-commutative'' is obviously weaker 
than the concept of independence. 
If $t_1$ and $t_2$ are independent, it follows that 
$\forall \psi~ \forall s \bullet
\langle s, t_1~\uparrow~t_2, \psi \rangle \wedge 
\langle s, t_2~\uparrow~t_1, \psi \rangle$.
Also note that, in contrast to the relation of transition independence, 
the ``semi-commutative'' relation is \emph{not symmetric}.

We now introduce a new definition for \emph{persistent set}.

\begin{definition}[Persistent Set Of A State]
A set $T \subseteq \mathcal{T}$ of transitions schedulable in a state $s \in \typesymbstate$ is
\emph{persistent in $s$} wrt. a property $\safety$ \emph{\myiff}, for all 
derivations
$s \stackrel{t_1}{\rightarrow} s_1 \stackrel{t_2}{\rightarrow}  s_2 \ldots 
\stackrel{t_{m-1}}{\rightarrow} s_{m-1} \stackrel{t_{m}}{\rightarrow} s_{m}$
including only transitions $t_i \in \mathcal{T}$ and $t_i \not\in T, 1\leq i \leq m$, 
each transition in T \emph{semi-commutes} with $t_i$ after $s$ wrt. $\sqsupseteq_{\psi}$.
\end{definition}

\begin{wrapfigure}{r}{0.55\textwidth}
\begin{center}
\vspace{-4mm}
\pppreset
\small{
$\begin{array}{rl}
\multicolumn{2}{l}{\mbox{Safety property $\safety$ and current state $s$}} \\
\ppp & T ~\Assign ~\emptyset \\
\ppp & \mbox{Add an enabled transition}~t~\mbox{into}~T \\
\ppp & \Foreach~\mbox{remaining schedulable transition}~t_i\\
\ppp & \ptab \If ~ \neg (\forall~ tp_{j} \in T \bullet \langle s, tp_j~\uparrow~t_i, \psi \rangle) \\
\ppp & \ptab \ptab \mbox{Add}~t_i~\mbox{into}~T\\ 
\end{array}$
}
\end{center}
\vspace{-4mm}
\caption{Computing a Persistent Set of a State}
\label{por:algo:persistent}
\end{wrapfigure}

For each state, computing a persistent set from the ``semi-commutative'' 
relation is similar to computing the classical persistent set 
under the transition independence relation. 
The algorithms for this task can be easily adapted from the algorithms 
presented in \cite{godefroid96}.
For convenience, we show one of such possibilities in 
Fig.~\ref{por:algo:persistent}.

As in traditional algorithms, the quality (i.e. the size) of the returned persistent
set is highly dependent on the first transitioned $t$ to be added and the order
in which the remaining transitions $t_i$ are considered. This is, however,
not the topic of the current paper.

With the new definition of persistent set, 
we now can proceed with the normal \emph{selective search} as described in 
classic \por~techniques. In the algorithm presented in 
Fig.~\ref{por:algo:trace}, we
perform depth first search (DFS), and only accommodate safety
verification (invariant property $\safety$).

\begin{theorem}\label{por:theorem:trace}
The selective search algorithm in Fig.~\ref{por:algo:trace} is 
\emph{sound}\footnote{Proof outline is in the Appendix.}. 
\qed
\end{theorem}

\begin{wrapfigure}{r}{0.51\textwidth}
\begin{center}
\vspace{-4mm}
\pppreset
\small{
$\begin{array}{rl}
\multicolumn{2}{l}{\mbox{Safety property $\safety$ and initial state $s_0$}} \\
\ppp &Initially: \func{Explore}(s_0) 				\\
\multicolumn{2}{l}{\mbox{\textbf{function}} ~ \func{Explore}(s)} \\
\ppp & 	\If~s \not\models \psi  ~\texttt{Report Error and TERMINATE} \\
\ppp &  T ~\Assign~ \func{Persistent\_Set}(s)			\\
\ppp &  \Foreach~\mbox{enabled transition}~t~in~T~\Do 		\\
\ppp &  \ptab \transition{\symstate}{\symstate'}{t} 
\ptab \ptab {\mbox{/* Execute t */}}		\\
\ppp & \ptab \func{Explore}(s')				\\	
\multicolumn{2}{l}{\mbox{\textbf{end function}}}
\end{array}$
}
\end{center}
\vspace{-4mm}
\caption{New Selective Search Algorithm}
\label{por:algo:trace}
\end{wrapfigure}

In preparing for \por~and \si~to work together,
we now further modify the concept of persistent set
so that it applies for a set of states sharing the same program point.
We remark that the previous definitions apply only for a specific state.
The key intuition is to attach a precondition $\phi$ to the program point of interest, 
indicating when semi-commutativity happens.

\begin{definition}[Semi-Commutative After A Program Point]
For a given concurrent program, a safety property $\safety$,
and  $t_1, t_2 \in \mathcal{T}$,
we say $t_1$ \emph{semi-commutes} with $t_2$ 
\emph{after program point} $\pc$ wrt. 
$\sqsupseteq_\psi$ and $\cons$,
denoted as $\langle \pc, \cons, t_1 \uparrow t_2, \psi \rangle$, 
\emph{\myiff}
for all state  $s \equiv \pair{\pc}{\unknown}$ 
reachable from the initial state 
$s_0$,
if $\symstate \models \cons$ then $t_1$ semi-commutes 
with $t_2$  after state s wrt. $\sqsupseteq_\psi$.
\end{definition}

\begin{definition}[Persistent Set Of A Program Point]
A set $T \subseteq \mathcal{T}$ of transitions \emph{schedulable} at 
program point $\pc$ is \emph{persistent at $\pc$} under the interpolant 
$\Intpsymbol$ wrt. a property $\safety$ \emph{\myiff}, 
for all state  $s \equiv \pair{\pc}{\unknown}$ 
reachable from the initial state 
$s_0$,
if $s \models \Intpsymbol$ then  
for all 
derivations
$s \stackrel{t_1}{\rightarrow} s_1 \stackrel{t_2}{\rightarrow}  s_2 \ldots 
\stackrel{t_{m-1}}{\rightarrow} s_{m-1} \stackrel{t_{m}}{\rightarrow} s_{m}$
including only transitions $t_i \in \mathcal{T}$ and $t_i \not\in T, 1\leq i \leq m$, 
each transition in T semi-commutes with $t_i$ after state $s$ wrt. $\sqsupseteq_\psi$.
\end{definition}

Assume that $T = \{tp_1, tp_2, \cdots tp_k\}$.
The interpolant $\Intpsymbol$ can now be computed as
$\Intpsymbol = \bigwedge \cons_{ji}$ for $1 \leq j \leq k, 1 \leq i \leq m$ 
such that $\langle \pc, \cons_{ji}, tp_j~\uparrow~t_i, \psi \rangle$.

For each program point, it is possible to 
have different persistent sets associated with 
different interpolants. In general,
a state which satisfies a stronger interpolant will have a 
smaller persistent set, therefore, 
it enjoys more pruning.



\section{Synergy of PDPOR and SI}
\label{sec:por:synergy}
We now show our combined framework.
We assume for each program point, a persistent set and its associated interpolant 
are computed statically, 
i.e., by separate analyses.
In other words, when we are at a program point, we can right away make use of 
the information about its persistent set.


The algorithm is in Fig.~\ref{por:algo:static}. 
The function $\func{Explore}$ has input $s$ and 
assumes the safety property at hand is $\safety$.
It naturally performs a depth first search of the state space.

\begin{wrapfigure}{r}{0.54\textwidth}
\begin{center}
\vspace{-2mm}
\pppreset
\small{
$\begin{array}{rl}
\multicolumn{2}{l}{\mbox{Assume safety property $\safety$ and initial state $s_0$}} \\
\ppp &Initially: \func{Explore}(s_0)				\\
\multicolumn{2}{l}{\mbox{\textbf{function}} ~ \func{Explore}(s)} \\
& \mbox{Let $\symstate$ be $\pair{\pc}{\unknown}$} \\
\ppp \label{por:algo:base1} & \If~(\func{memoed}(s, \Intpsymbol))~\Return~\Intpsymbol		\\
\ppp \label{por:algo:base2} & \If~(\symstate \not\models \safety)  
~\texttt{Report Error and TERMINATE}\\
\ppp &  \Intpsymbol ~\Assign~ \safety						\\

\ppp &  \langle T, \Intpsymbol_{trace} \rangle ~\Assign~ \func{Persistent\_Set}(\pc)			\\
\ppp & \If~(\symstate \models \Intpsymbol_{trace}) \\
\ppp & \ptab \func{Ts}~\Assign~T \\
\ppp\label{por:algo:combine_trace} & \ptab \Intpsymbol ~\Assign~\Intpsymbol~\wedge~\Intpsymbol_{trace} \\
\ppp & \Else~\func{Ts}~\Assign~Schedulable(s) \\

\ppp &  \Foreach~t~in~(\func{Ts}~\setminus~Enabled(s))~\Do	\\
\ppp \label{por:algo:disabled}& \ptab \ptab \Intpsymbol ~\Assign~\Intpsymbol~\wedge~
\pre{t}{\false}	\\

\ppp \label{por:algo:enabled-persistent} &  \Foreach~t~in~(\func{Ts}~\cap~Enabled(s))~\Do	 		\\
\ppp \label{por:algo:extend}&  \ptab \transition{\symstate}{\symstate'}{t}	
\ptab \ptab \ptab \ptab  {\mbox{/* Execute t */}}		\\
\ppp \label{por:algo:recursive}&  \ptab \Intpsymbol' ~\Assign~ \func{Explore}(s')			\\
\ppp \label{por:algo:child}&  \ptab \Intpsymbol ~\Assign ~ \Intpsymbol~\wedge~
\pre{t}{\Intpsymbol'}		\\	
\ppp\label{por:algo:return}&  \func{memo}~$and$~\Return~(\Intpsymbol)					\\
\multicolumn{2}{l}{\mbox{\textbf{end function}}}
\end{array}$
}
\end{center}
\vspace{-10pt}
\caption{A Framework for POR and SI (DFS)} 
\label{por:algo:static}
\end{wrapfigure}

\noindent \textbf{Two Base Cases:} The function $\func{Explore}$ handles two base cases. 
One is when the current state is subsumed by some computed 
(and memoed) interpolant $\Intpsymbol$.
No further exploration is needed, and $\Intpsymbol$ is returned as the interpolant 
(line~\ref{por:algo:base1}). 
The second base case is when the current state is found to be \emph{unsafe}
(line~\ref{por:algo:base2}).

\vspace{2pt}
\noindent \textbf{Combining Interpolants:} 
We make use of the (static) persistent set $T$ computed for the current program point.
We comment further on this in the next section. 

The set of transitions to be considered is denoted by \func{Ts}.
When the current state implies the interpolant $\Intpsymbol_{trace}$ associated with $T$,
we need to consider only those transitions in $T$. Otherwise, we need to consider
all the schedulable transitions. 
Note that when the persistent set $T$ is employed,
the  interpolant $\Intpsymbol_{trace}$ must contribute to the combined interpolant 
of the current state (line~\ref{por:algo:combine_trace}).
Disabled transitions at the current state will strengthen the interpolant 
as in line~\ref{por:algo:disabled}.
Finally, we recursively follow those transitions which 
are enabled at the current state.
The interpolant of each child state contributes to the interpolant of the current state as
in line~\ref{por:algo:child}. In our framework, interpolants
are propagated back using the precondition operation \wpc,
where $\pre{t}{\phi}$ denotes a \emph{safe approximation} of the weakest precondition
wrt. the transition $t$ and the postcondition $\phi$ \cite{dijkstrawp}.



\begin{theorem}
The algorithm in Fig.~\ref{por:algo:static} is 
\emph{sound}\footnote{Proof outline is in the Appendix.}. 
\qed
\end{theorem}
\section{Implementation of PDPOR}
\label{sec:por:implementation}
We now elaborate on the remaining task: how to
estimate the semi-commutative relation, thus deriving
the (static) persistent set at a program point. 
Similar to the formalism of traditional \por{}, our formalism is
of paramount importance for the semantic use as well as to construct 
the formal proof of correctness (see the Appendix).
In practice, however, we have to come up with sufficient conditions 
to efficiently implement the concepts.
In this paper, we estimate the semi-commutative relation in two steps: 

\vspace{-2mm}

\begin{enumerate}

\item
We first employ \emph{any} traditional \por~method
and first estimate the ``semi-commutative'' relation as
the traditional independence relation 
(then the corresponding condition $\phi$ is just $\true$). 
This is possible because the proposed concepts are \emph{strictly weaker} than the
corresponding concepts used in traditional \por~methods.

\item 
We then identify and exploit a number of patterns under which 
we can statically derive and prove  
the semi-commutative relation between transitions. 
In fact, these simple patterns suffice to deal with 
a number of important real-life applications.

\end{enumerate}

\vspace{-2mm}

\noindent
In the rest of this section, we outline
three common classes of problems,
from which the semi-commutative relation 
between transitions can be easily identified and proved, i.e., our step
2 becomes applicable.

\vspace{2pt}
\noindent
\textbf{\underline{Resource Usage of Concurrent Programs:}}
Programs make use of limited resource (such as time, memory, bandwidth).
Validation of resource usage in sequential setting is already a hard problem. 
It is obviously more challenging in the setting of concurrent programs
due to process interleavings.

Here we model this class of problems by using a resource variable $\resource$. 
Initially, $\resource$ is \emph{zero}. 
Each process can increment or decrement variable
$\resource$ by some concrete value (e.g., memory allocation or deallocation respectively).
A process can also double the value $\resource$ 
(e.g., the whole memory is duplicated).
However, the resource variable $\resource$ 
cannot be used in the guard condition of any
transition, i.e., we cannot model the behavior of
a typical garbage collector.
The property to be verified is that, ``at all times, $\resource$ is (upper-) 
bounded by some constant''.

\begin{proposition}\label{por:prop1}
Let $\resource$ be a  resource variable of a concurrent program, 
and assume the safety property at hand is
$\psi \equiv \resource \le C$, where C is a constant. 
For all transitions (assignment operations only)
$t_1: \resource = \resource + c_1$, $t_2: \resource = \resource * 2$, 
$t_3: \resource = \resource - c_2$
where $c_1, c_2 > 0$, we have for all program points $\pc$:\\
\indent
$\langle \pc, \true, t_1~\uparrow~t_2, \psi \rangle \wedge 
\langle \pc, \true, t_1~\uparrow~t_3, \psi \rangle \wedge 
\langle \pc, \true, t_2~\uparrow~t_3, \psi \rangle$ \qed
\end{proposition}

Informally, other than common mathematical facts such as
additions can commute and so do multiplications and subtractions, 
we also deduce that additions can semi-commute with both multiplications and subtractions while
multiplications can semi-commute with subtractions. 
This Proposition can be proved by using basic laws of algebra.


\newexample{}
Let us refer back to the example of two closely coupled processes introduced in 
Sec.~\ref{sec:por:prelim}, but now under the assumption that $x$ is the resource variable of interest.
Using the semi-commutative relation derived from Proposition~\ref{por:prop1}, we need to explore only 
\emph{one complete trace} to prove this safety.


We recall that, in contrast, \por~(and {\small \textsf{DPOR}})-only methods 
will enumerate the full execution tree which
contains 19 states and 6 complete execution traces.
Any technique which employs only the notion of Mazurkiewicz trace equivalence for pruning 
will have to consider all 6 complete traces (due to 6 different terminal states).
\si~alone can reduce the search space in this example, and requires to explore 
only 9 states and 4 subsumed states (as in Sec.~\ref{sec:por:prelim}).

\vspace{2pt}
\noindent
\textbf{\underline{Detection of Race Conditions:}}
\cite{wang08atva} proposed a property driven pruning algorithm to 
detect race conditions in multithreaded programs.
This work has achieved more reduction in comparison with \dpor{}.
The key observation is that, at a certain location (program point) $\pc$,
if their conservative ``lockset analysis'' shows that a search subspace is race-free,
the subspace can be pruned away. 
As we know, \dpor{} relies solely on the independence  relation to prune redundant interleavings
(if $t_1$, $t_2$ are independent, there is no need to flip their execution order).
In \cite{wang08atva}, however, even when $t_1$, $t_2$ are dependent, 
we may skip the corresponding search space 
if flipping the order of $t_1$, $t_2$ does not affect the reachability 
of any race condition. 
In other words, \cite{wang08atva} is indeed a (conservative) realization of 
our \pdpor, specifically targeted for detection of race conditions.
Their mechanism to capture the such scenarios is by introducing 
a trace-based lockset analysis.

\vspace{2pt}
\noindent
\textbf{\underline{Ensuring Optimistic Concurrency:}}
In the implementations of many concurrent protocols, \emph{optimistic concurrency},
i.e., at least one process commits, is usually desirable.
This can be modeled by introducing a \textsf{flag} variable which will
be set when some process commits. The \textsf{flag} variable once
set can not be unset. 
It is then easy to see that for all program point $\pc$ and
transitions $t_1, t_2$, we have 
$\langle \pc, \textsf{flag = 1}, t_1~\uparrow~t_2, \psi \rangle$.
Though simple, this observation will bring us more reduction compared
to traditional \por{} methods.

\section{Experiments}
\label{sec:por:experiment}

This section conveys two key messages. First,
when trace-based and state-based  methods
are not effective individually, our combined framework still
offers significant reduction. Second, property driven \por{} can be very
effective, and applicable not only to academic programs, but
also to programs used as benchmarks in the state-of-the-art.

We use a 3.2 GHz Intel processor and 2GB memory running Linux. 
Timeout is set at 10 minutes.  
In the tables, cells with `-' indicate timeout.
We compare the performance of Partial Order Reduction alone (\por), State Interpolation alone
(\si), the synergy of Partial Order Reduction and State Interpolation  
({\small \textsf{POR+SI}}), i.e., the semi-commutative relation is estimated using 
only step 1 presented in Sec.~\ref{sec:por:implementation}, and when applicable, the 
synergy of Property Driven Partial Order Reduction and State Interpolation  ({\small \textsf{PDPOR+SI}}),
i.e., the semi-commutative relation is estimated using both
steps presented in Sec.~\ref{sec:por:implementation}.
For the \por{} component, we use the implementation from \cite{ase11kobor}.

Table~\ref{por:table:main} 
starts with parameterized versions of the \emph{producer/consumer} example
because its basic structure is extremely common.
There are $2 * N$ producers and 1 consumer.  
Each producer will do its own non-interfered computation
first, modeled by a transition which does not interfere with other
processes. Then these producers will modify the shared variable
$x$ as follows: each of the first $N$ producers increments $x$, while
the other $N$ producers double the value of $x$.  On the other hand, the
consumer consumes the value of $x$.  The safety property is that
the consumed value is no more than $N * 2^N$.

Table~\ref{por:table:main} 
clearly demonstrates the synergy benefits
of \por~and \si. {\small \textsf{POR+SI}} 
significantly outperforms both \por~and \si.
Note that this example can easily be translated to the resource usage problem, 
where our \pdpor{} requires only a {\em single} trace (and less than 0.01 second)
in order to prove safety.

\begin{wraptable}{r}{0.62\textwidth}
\centering{
\small{
\caption{Synergy of \por{} and \si{}}
\begin{tabular}{|c|r|r||r|r||r|r|}
  \cline{1-7} 
  & \multicolumn{2}{|c||}{\textsf{POR}} & \multicolumn{2}{|c||}{\textsf{SI}} & 
  \multicolumn{2}{|c|}{\textsf{POR+SI}} \\ 
\cline{1-7} 
 
  \textsf{Problem} &  \textsf{States} & \textsf{T(s)} & \textsf{States} & 
  \textsf{T(s)} & \textsf{States} &  \textsf{T(s)} \\ 
 \cline{1-7} 
  \texttt{p/c-2}    & $449$ & $0.03$ & $514$ & $0.17$ & $85$ & $0.03$ \\
  \texttt{p/c-3}   & $18745$ & $2.73$ & $6562$ & $2.43$ & $455$ & $0.19$ \\ 
  \texttt{p/c-4}    & $986418$ & $586.00$ & $76546$ & $37.53$ & $2313$ & $1.07$ \\ 
  \texttt{p/c-5}     & $-$ & $-$ & $-$ &  $-$ & $11275$ & $5.76$ \\ 
  \texttt{p/c-6}       & $-$  & $-$ & $-$ & $-$ & $53261$ & $34.50$ \\ 
  \texttt{p/c-7}       & $-$  & $-$ & $-$ & $-$ & $245775$ & $315.42$ \\ 
 \hline
    \hline  
  \texttt{din-2a}    &  $22$ & $0.01$ & $21$ & $0.01$ & $21$ & $0.01$  \\
  \texttt{din-3a}     & $646$ & $0.05$ & $153$ & $0.03$ & $125$ & $0.02$ \\
  \texttt{din-4a}    & $155037$ & $19.48$ & $1001$ &$0.17$ & $647$ & $0.09$ \\
  \texttt{din-5a}     & $-$ & $-$ & $6113$  &$1.01$  & $4313$ & $0.54$ \\
  \texttt{din-6a}      & $-$ & $-$ & $35713$ &$22.54$  & $24201$ & $4.16$ \\
  \texttt{din-7a}      & $-$ & $-$ & $202369$ & $215.63$  & $133161$ & $59.69$ \\
 \hline
 \hline
  \texttt{bak-2}     & $48$ & $0.03$ & $38$ & $0.03$ & $31$ & $0.02$ \\
  \texttt{bak-3}     & $1003$ & $1.85$ & $264$ & $0.42$ & $227$ & $0.35$ \\
  \texttt{bak-4}    & $27582$ & $145.78$ & $1924$ & $5.88$ & $1678$ & $4.95$  \\
  \texttt{bak-5}    & $-$ & $-$ & $14235$ & $73.69$ & $12722$ & $63.60$ \\
\hline

\end{tabular}
\label{por:table:main}
}}
\end{wraptable}

We next use the parameterized version of the \emph{dining philosophers}.
We chose this for two reasons.  First,
this is a classic example often used in concurrent algorithm design to
illustrate synchronization issues and techniques for resolving them.
Second, previous work~\cite{nec09cav}
has used this to
demonstrate benefits from combining \por~and \smt{}.

The first safety property used in \cite{nec09cav}, ``it is not
that all philosophers can eat simultaneously'', is somewhat trivial.
Therefore, here we verify a \emph{tight} property,  
which is (a): ``no more than \emph{half} the philosophers can eat simultaneously''.
To demonstrate the power of symbolic execution,
we verify this property \emph{without} knowing the initial configurations 
of all the forks.
Table~\ref{por:table:main}, again, demonstrates the
significant improvements of {\small \textsf{POR+SI}} over \por{} alone and \si{} alone.
We note that the performance of our {\small \textsf{POR+SI}} algorithm
is about 3 times faster than \cite{nec09cav}.

We additionally considered a second safety property as in \cite{nec09cav}, namely (b):
``it is possible to reach a state in which all philosophers have eaten at least once''.
Our symbolic execution framework requires only a \emph{single
trace} (and less than $0.01$ second) to prove this property in all instances, 
whereas \cite{nec09cav} requires even more time
compared to proving property (a). This illustrates
the scalability issue of \cite{nec09cav}, which is 
representative for other techniques employing  general-purpose 
\smt~solver for symbolic pruning.

We also perform experiments on the ``Bakery'' algorithm.
Due to existence of infinite domain variables, 
model checking hardly can handle this case. 
Here we remark that in symbolic methods, loop handling is often considered
as an orthogonal issue.
Programs with statically bounded loops can be easily unrolled into 
equivalent loop-free programs.
For unbounded loops, either loop invariants are provided or the 
employment of some invariant discovery routines, e.g., \cite{jaffar11rv},
is necessary. In order for our algorithm to work here,
we make use of the standard loop invariant for this example.

To further demonstrate the power our synergy framework
over \cite{nec09cav}\footnote{\cite{nec09cav}
is not publicly available. Therefore, it is not possible for us to 
make more comprehensive comparisons.}  as well as the power of our property driven \por{},
we experiment next on the \emph{Sum-of-ids} program.
Here, each process (of $N$ processes) has one unique \emph{id} and will
increment a shared variable \emph{sum} by this \emph{id}.
We prove that in the end this variable will be incremented by 
the sum of all the ids.

\begin{wraptable}{r}{0.6\textwidth}
\centering{
\small{
\caption{Comparison with \cite{nec09cav}}
\begin{tabular}{|r|r|r|r||r|r||r|r|}
  \hline 
  & \multicolumn{3}{|c||}{\textsf{\cite{nec09cav} w. Z3}}
   &  \multicolumn{2}{|c||}{\textsf{POR+SI}}        &  \multicolumn{2}{|c|}{\textsf{PDPOR+SI}} \\
  \hline
  \textsf{T(s)} & \textsf{\#C} & \textsf{\#D} & \textsf{T(s)} & \textsf{States} & 
  \textsf{T(s)}  & \textsf{States} & \textsf{T(s)} \\ 
  \hline  
  \texttt{sum-6}   & $1608$ & $1795$ & $0.08$ & $193$ & $0.05$ & $7$ & $0.01$ \\
  \texttt{sum-8}     & $54512$ & $59267$ & $10.88$ & $1025$ & $0.27$  & $9$ & $0.01$ \\
  \texttt{sum-10}    & $-$ & $-$ & $-$ & $5121$ & $1.52$  & $11$ & $0.01$ \\
  \texttt{sum-12}     & $-$ & $-$ & $-$ & $24577$ & $8.80$ & $13$ & $0.01$ \\
  \texttt{sum-14}     & $-$ & $-$ & $-$ & $114689$ & $67.7$ &  $15$ & $0.01$ \\
  \hline

\end{tabular}
\label{por:table:sum}
}
}
\end{wraptable}

See Table~\ref{por:table:sum}, where we experiment with Z3~\cite{z3} (version 4.1.2)
using the encodings presented in \cite{nec09cav}.
\textsf{\#C} denotes the
number of conflicts while \textsf{\#D} denotes the
number of decisions made by Z3.
We can see that our synergy framework scale much better than \cite{nec09cav} with Z3.
Also, this example can also be translated to resource usage problem, 
our use of property-driven \por{} again requires \emph{one} single trace to prove safety.

\begin{wraptable}{r}{0.55\textwidth}
\centering
\small{
\caption{Experiments on  \cite{ICSE11}'s Programs}
\begin{tabular}{|l|r|r||r|r||r|r|}
  \hline 
    & \multicolumn{2}{|c||}{\textsf{\cite{ICSE11}}} & \multicolumn{2}{|c||}{\textsf{SI}} 
                                          & \multicolumn{2}{|c|}{\textsf{PDPOR+SI}} \\
  \hline
\textsf{Problem}  &  \textsf{C} & \textsf{T(s)} & \textsf{States} &  \textsf{T(s)} 
& \textsf{States} &  \textsf{T(s)} \\ 
  \hline  
  \textsf{micro\_2}    & $17$ & $1095$ & $20201$ & $10.88$ & $201$ & $0.04$   \\
  \textsf{stack}    & $12$ & $225$ & $529$ & $0.26$ & $529$ & $0.26$         \\
  \textsf{circular\_buffer}    & $\infty$ & $477$ & $29$ & $0.03$ & $29$ & $0.03$  \\
  \textsf{stateful20}    & $10$ & $95$ & $1681$ & $1.13$ & $41$ & $0.01$  \\
  
 \hline
\end{tabular}
\label{por:table:smt}
}
\end{wraptable}

Finally, to benchmark our framework with \smt-based methods,
we select four \emph{safe} programs from \cite{ICSE11} 
where the experimented methods did not perform well.
Those programs are \textsf{micro\_2}, \textsf{stack}, \textsf{circular\_buffer},
and \textsf{stateful20}. We note that safe programs allow fairer 
comparison between different approaches since to verify them we have to cover
the whole search space.
Table~\ref{por:table:smt} shows the running time of \si{} alone and
of the combined framework. For convenience, we also
tabulate the \emph{best} running time reported in \cite{ICSE11} 
and \textsf{C} is the context switch bound used.
We assume no context switch bound, hence
the corresponding value in our framework is $\infty$.

We can see that even our \si{} alone significantly 
outperforms the techniques in \cite{ICSE11}.
We believe it is due to the inefficient encoding of process interleavings
(mentioned in Sec.~\ref{sec:por:related}) as well as the following reasons.
First, our method is \emph{lazy}, 
which means that only a path is considered at a time:
\cite{ICSE11} itself demonstrates partially the usefulness of this.
Second, but importantly, we are \emph{eager} in discovering infeasible paths.
The program \textsf{circular\_buffer}, which has only
one feasible complete execution trace, can be efficiently handled by our framework,
but not \smt{}.
This is one important advantage of our symbolic execution framework over
\smt{}-based methods, as discussed in
\cite{mcmillan10cav}.

It is important to note that, \pdpor{} significantly
improves the performance of \si{} wrt. programs \textsf{micro\_2}
and \textsf{stateful20}. This further demonstrates the applicability
of our proposed framework.


\section{Conclusion}
\label{sec:por:conclusion}
We present a concurrent verification framework which synergistically combines
trace-based reduction techniques with the recently established notion
of \emph{state interpolant}.
One key contribution is the new concept of property-driven \por{} which serves to reduce
more interleavings than previously possible. We believe that fully automated techniques
to compute the semi-commutative relations for \emph{specific}, but important, 
application domains will be interesting future work.


\scriptsize{
\bibliographystyle{plain}
\bibliography{references}

\begin{thebibliography}{10}

\bibitem{abdulla14popl}
P.~Abdulla, S.~Aronis, B.~Jonsson, and K.~Sagonas.
\newblock Optimal dynamic partial order reduction.
\newblock In {\em POPL}, 2014.

\bibitem{alur97cav}
R.~Alur, R.~K. Brayton, T.~A. Henzinger, S.~Qadeer, and S.~K. Rajamani.
\newblock Partial-order reduction in symbolic state space exploration.
\newblock In {\em CAV}, 1997.

\bibitem{ase11kobor}
P.~Bokor, J.~Kinder, M.~Serafini, and N.~Suri.
\newblock Supporting domain-specific state space reductions through local
  partial-order reduction.
\newblock In {\em ASE}, 2011.

\bibitem{cadar11icse}
C.~Cadar, P.~Godefroid, S.~Khurshid, C.~S. P\u{a}s\u{a}reanu, K.~Sen,
  N.~Tillmann, and W.~Visser.
\newblock Symbolic execution for software testing in practice: preliminary
  assessment.
\newblock In {\em ICSE}, 2011.

\bibitem{ICSE11}
L.~Cordeiro and B.~Fischer.
\newblock Verifying multi-threaded software using smt-based context-bounded
  model checking.
\newblock In {\em ICSE}, 2011.

\bibitem{z3}
L.~De~Moura and N.~Bj{\o}rner.
\newblock Z3: an efficient smt solver.
\newblock In {\em TACAS}, 2008.

\bibitem{dijkstrawp}
E.~W. Dijkstra.
\newblock Guarded commands, nondeterminacy and formal derivation of programs.
\newblock {\em Commun. ACM}, 1975.

\bibitem{flanagan05popl}
C.~Flanagan and P.~Godefroid.
\newblock Dynamic partial-order reduction for model checking software.
\newblock In {\em POPL}, 2005.

\bibitem{godefroid96}
P.~Godefroid.
\newblock {\em Partial-Order Methods for the Verification of Concurrent
  Systems: An Approach to the State-Explosion Problem}.
\newblock Springer-Verlag New York, Inc., 1996.

\bibitem{grumberg05popl}
O.~Grumberg, F.~Lerda, O.~Strichman, and M.~Theobald.
\newblock Proof-guided underapproximation-widening for multi-process systems.
\newblock In {\em POPL}, 2005.

\bibitem{cartesian}
G.~Gueta, C.~Flanagan, E.~Yahav, and M.~Sagiv.
\newblock Cartesian partial-order reduction.
\newblock In {\em SPIN}, 2007.

\bibitem{jaffar11rv}
{J}. {J}affar, {J}.{A}. {N}avas, and {A}. {S}antosa.
\newblock {U}nbounded {S}ymbolic {E}xecution for {P}rogram {V}erification.
\newblock In {\em RV}, 2011.

\bibitem{jaffar09cp}
J.~Jaffar, A.~E. Santosa, and R.~Voicu.
\newblock An interpolation method for clp traversal.
\newblock In {\em CP}, 2009.

\bibitem{nec09cav}
V.~Kahlon, C.~Wang, and A.~Gupta.
\newblock Monotonic partial order reduction: An optimal symbolic partial order
  reduction technique.
\newblock In {\em CAV}, 2009.

\bibitem{king76acm}
J.~C. King.
\newblock {S}ymbolic {E}xecution and {P}rogram {T}esting.
\newblock {\em Com. ACM}, 1976.

\bibitem{mazurkiewicz86}
A.~W. Mazurkiewicz.
\newblock Trace theory.
\newblock In {\em Advances in Petri Nets}, 1986.

\bibitem{mcmillan10cav}
K.~L. McMillan.
\newblock Lazy annotation for program testing and verification.
\newblock In {\em CAV}, 2010.

\bibitem{peled93cav}
D.~Peled.
\newblock All from one, one for all: on model checking using representatives.
\newblock In {\em CAV}, 1993.

\bibitem{grasp}
J.~P.~M. Silva and K.~A. Sakallah.
\newblock Grasp-a new search algorithm for satisfiability.
\newblock In {\em ICCAD}, 1996.

\bibitem{valmari89petri}
A.~Valmari.
\newblock Stubborn sets for reduced state space generation.
\newblock In {\em Applications and {T}heory of {P}etri {N}ets}, 1989.

\bibitem{wachter13fmcad}
B.~Wachter, D.~Kroening, and J.~Ouaknine.
\newblock Verifying multi-threaded software with impact.
\newblock In {\em FMCAD}, 2013.

\bibitem{nec09fse}
C.~Wang, S.~Chaudhuri, A.~Gupta, and Y.~Yang.
\newblock Symbolic pruning of concurrent program executions.
\newblock In {\em ESEC/FSE}, 2009.

\bibitem{wang08atva}
C.~Wang, Y.~Yang, A.~Gupta, and G.~Gopalakrishnan.
\newblock Dynamic model checking with property driven pruning to detect race
  conditions.
\newblock In {\em ATVA}, 2008.

\bibitem{nec08tacas}
C.~Wang, Z.~Yang, V.~Kahlon, and A.~Gupta.
\newblock Peephole partial order reduction.
\newblock In {\em TACAS}, 2008.

\bibitem{StatefulSPIN08}
Y.~Yang, X.~Chen, G.~Gopalakrishnan, and R.~M. Kirby.
\newblock Efficient stateful dynamic partial order reduction.
\newblock In {\em SPIN}, 2008.

\end{thebibliography}
}

\newpage
\appendix

\normalsize

\section*{Appendix: Proofs of the Two Theorems}\label{sec:app:proof}
\setcounter{theorem}{0}
\begin{theorem}\label{por:theorem:trace}
The selective search algorithm in Fig.~\ref{por:algo:trace} is sound.
\end{theorem}

\begin{proof}[Outline]
Assume that there exist some traces which violate the property $\psi$ 
and are not examined by our \emph{selective search}.
Let denote the set of such traces as $\mathcal{W}_{violated}$.
For each trace $\rho =  s_0 \stackrel{t_1}{\rightarrow} s_1 \stackrel{t_2}{\rightarrow} s_2 
\cdots \stackrel{t_m}{\rightarrow} s_m$, $\rho \in \mathcal{W}_{violated}$, 
let $first(\rho)$ denote the smallest index $i$ such that $t_i$ 
is not in the persistent set of $s_{i-1}$. 
Without loss of generality, 
assume $\rho_{max} =  s_0 \stackrel{t_1}{\rightarrow} s_1 \stackrel{t_2}{\rightarrow} s_2 
\cdots \stackrel{t_m}{\rightarrow} s_m$ having the maximum $first(\rho)$.
Let $i = first(\rho_{max}) < m$.
As the ``commuter'' and ``commutee'' cannot ``de-schedule'' each other, 
in the set $\{t_{i+1} \cdots t_m\}$ there must be a transition 
which belongs to the persistent set of $s_{i-1}$ (otherwise, the must exist some transition 
that belongs to the persistent set of $s_{i-1}$ which is schedulable at $s_m$. 
Therefore $s_m$ is not a terminal state).
Let $j$ be the smallest index such that $t_j$ 
belongs to the persistent set of $s_{i-1}$. 
By definition, wrt. $\sqsupseteq_{\psi}$ and after $s_{i-1}$,  
$t_j$ \emph{semi-commutes} with $t_i, t_{i+1}, \cdots t_{j-1}$.
Also due to the definition of the ``semi-commutative'' relation 
we deduce that all the following traces 
(by making $t_j$ repeatedly commute backward):
\begin{center}
$\begin{aligned} 
\rho'_1 = t_1 t_2 \cdots t_{i-1} t_i & t_{i+1} \cdots t_j t_{j-1} t_{j+1} \cdots t_m \\
&\vdots \\ 
\rho'_{j-i-1} = t_1 t_2 \cdots t_{i-1} t_i & t_j t_{i+1} \cdots t_{j-1} t_{j+1} \cdots t_m \\
\rho'_{j-i} = t_1 t_2 \cdots t_{i-1} t_j & t_i t_{i+1} \cdots t_{j-1} t_{j+1} \cdots t_m \\
\end{aligned}$ 
\end{center}

\noindent
must violate the property $\psi$ too. 
However, $first(\rho'_{j-i}) >$ $first(\rho_{max})$. 
This contradicts the definition of $\rho_{max}$.
\qed
\end{proof}

\begin{figure}[bh!]
\centering
\includegraphics[width=0.9\textwidth]{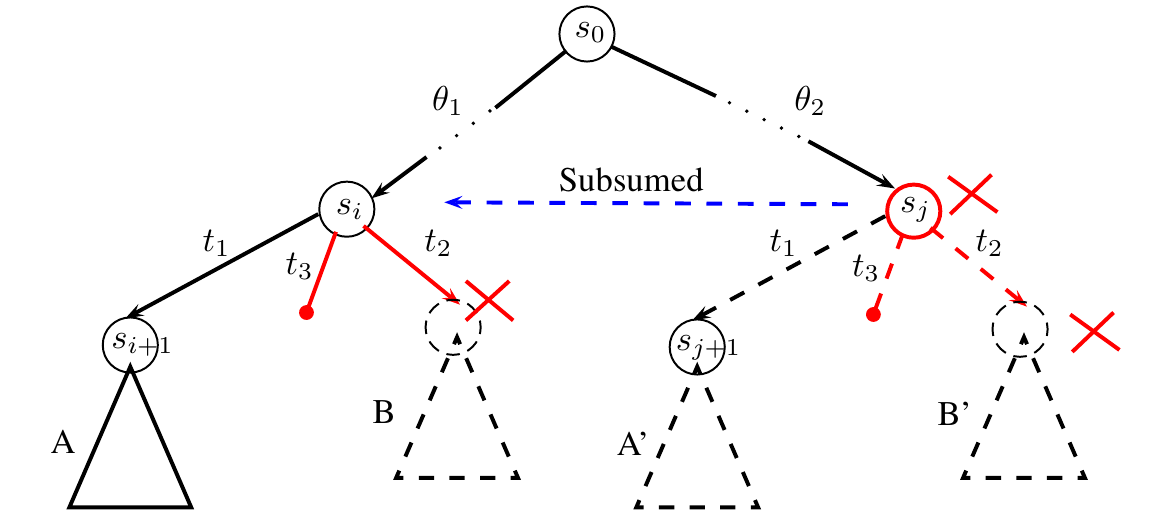}
\caption{Inductive Correctness}
\label{por:fig:proof}
\end{figure}

\begin{theorem}
The synergy algorithm in Fig.~\ref{por:algo:static} is sound.
\end{theorem}

\begin{proof}[Outline]
We use structural induction. 
Refer to Fig.~\ref{por:fig:proof}.
Assume that from $s_0$ we reach state $s_i \equiv \pair{\pc}{\unknown}$.
W.l.o.g., assume that 
at $s_i$ there are three transitions which are schedulable, namely $t_1, t_2, t_3$,  of which only
$t_1$ and $t_2$ are enabled. Also assume that under the interpolant $\Intpsymbol_2$,
the persistent set of $\pc$, and therefore of $s_i$, is just $\{t_1\}$. From the algorithm, we will extend
$s_i$ with $t_1$ (line~\ref{por:algo:extend}) and attempt to verify the subtree $A$ 
(line~\ref{por:algo:recursive}). 
Our induction hypothesis is 
that we have finished considering $A$, and indeed, it is safe under the interpolant $\Intpsymbol_A$.
That subtree will contribute $\Intpsymbol_1 = \pre{t_1}{\Intpsymbol_A}$ (line~\ref{por:algo:child}) 
to the interpolant of $s_i$.

Using the interpolant $\Intpsymbol_2$, the branch having transition $t_2$ 
followed by the subtree $B$
is pruned and that is safe, due to Theorem~\ref{por:theorem:trace}. 
Also, the disabled transition $t_3$ contributes
$\Intpsymbol_3$ (line~\ref{por:algo:disabled}) to the interpolant of $s_i$.
Now we need to prove that \textbf{$\Intpsymbol = \Intpsymbol_1 \wedge \Intpsymbol_2 \wedge \Intpsymbol_3$ is indeed a sound
interpolant for program point $\pc$}. 

Assume that subsequently in the search, 
we reach some state $s_j \equiv  \pair{\pc}{\unknown}$.
We will prove that $\Intpsymbol$ is a sound interpolant of $\pc$ by proving that if
$s_j \models \Intpsymbol$, then the pruning of $s_j$ is safe.

First, $s_j \models \Intpsymbol$ implies that $s_j \models \Intpsymbol_3$. 
Therefore, at $s_j$, $t_3$ is also disabled.
Second, assume that $t_1$ is enabled at $s_j$ and $\transition{s_j}{s_{j+1}}{t_1}$ 
(if not the pruning of $t_1$ followed by $A'$ is definitely safe).
Similarly, $s_j \models \Intpsymbol$ implies that 
$s_j \models \Intpsymbol_1$. 
Consequently, $s_{j+1} \models \Intpsymbol_A$ and
therefore the subtree $A'$ is safe too.
Lastly, $s_j \models \Intpsymbol$ implies that $s_j \models \Intpsymbol_2$. 
Thus the reasons which ensure that
the traces ending with subtree $A$ cover the traces ending with subtree $B$ also hold at $s_j$.
That is, the traces ending with subtree $A'$ 
also cover the traces ending with subtree $B'$.
\qed
\end{proof}


\end{document}